\documentclass[12pt]{amsart}
\usepackage{amscd,amsthm,amsfonts,amssymb,amsmath,euscript}
\usepackage[dvips]{graphicx}
\usepackage[dvips]{graphics}
\usepackage[T2A]{fontenc}
\usepackage[cp866]{inputenc}

\newtheorem{theorem}{Theorem}[section]
\newtheorem{lemma}{Lemma}[section]
\newtheorem{corollary}{Corollary}[section]

\newcommand{\Aut}{\mathop{\sf Aut}\nolimits}

\newcommand{\tr}{\mathop{\sf tr}\nolimits}

\,

\,

\,

\,

\,

\,

\,

\,

\,

\,

\,

\,

\,

\,

\,

\,

\,

\,

\,

\,

\,

\,

\,

\,

\,

\,

\,

\,

\,

\,

\,

\,

\,

\,

\,

\,

\,

\,

\,

\,

\,

\,

\,

\,

\,

\,

\,

\,

\,

\,

\,

\,

\,

\,

\,

\begin{document}

\begin{center}

\hfill FIAN/TD-29/12\\
\hfill ITEP/TH-55/12\\
\hfill IIP-TH-30/12

\vspace{1cm}

\end{center}

\centerline{\large\bf Asymptotic Hurwitz numbers}

\vspace{1cm}

\centerline {A.Mironov, A.Morozov, S.Natanzon}
\address{\begin{flushleft}Theory Department, Lebedev Physical Institute, Moscow, Russia;\end{flushleft}
\begin{flushleft}Institute for Theoretical and Experimental Physics, Moscow, Russia;\end{flushleft}
IIP, Federal University of Rio Grande do Norte, Natal, Brazil}
\email{mironov@itep.ru; mironov@lpi.ru}
\address{\begin{flushleft}Institute for Theoretical and Experimental Physics, Moscow, Russia;\end{flushleft} IIP,
Federal University of Rio Grande do Norte, Natal, Brazil}
\email{morozov@itep.ru}
\address {\begin{flushleft}Department of Mathematics, Higher School of Economics, Moscow, Russia;\end{flushleft}
\begin{flushleft}A.N.Belozersky Institute, Moscow State University, Russia;\end{flushleft} Institute
for Theoretical and Experimental Physics, Moscow, Russia} \email{natanzons@mail.ru}

\vspace{1cm}

\centerline{ABSTRACT}

\bigskip

{\footnotesize
The classical Hurwitz numbers of degree $n$ together with the Hurwitz numbers of the seamed surfaces of degree $n$
give rise to the Klein topological field theory \cite{AN2}. We extend this construction to the Hurwitz numbers of all
degrees at once. The corresponding Cardy-Frobenius algebra is induced by arbitrary Young diagrams and arbitrary
bipartite graphs. It turns out to be isomorphic to the algebra of differential operators from
\cite{MMN3} which serves a model for open-closed string theory. The operator associated with the Young diagram
of the transposition of two elements
coincides with the cut-and-join operator which gives rise to relations for the classical Hurwitz numbers.
We prove that the operators corresponding to arbitrary Young diagrams and bipartite graphs also give rise to
relations for the Hurwitz numbers.
 }

\tableofcontents

\section{Introduction}

In accordance with \cite{AN2}, the classical Hurwitz numbers of degree $n$ along with the Hurwitz numbers
of the seamed surfaces
\footnote{The seamed surfaces were accurately defined in \cite{Roz}. They are also called world-sheet foam \cite{KR}.}
of degree $n$ gives rise to an open-closed (and even Klein in the meaning of \cite{AN})
topological field theory. The corresponding Cardy-Frobenius algebra consists of the algebra of Young diagrams of degree $n$,
of the algebra of bipartite graphs of degree $n$ and a homomorphism of the first algebra to the second one.
It can be considered as a realization of open-closed string theory \cite{MMN4}.

In this paper, we extend this construction to the algebra of the Hurwitz numbers of all degrees. We prove that
this universal Hurwitz algebra is associated with an infinite-dimensional topological field theory, and the corresponding
Cardy-Frobenius algebra is isomorphic to the Cardy-Frobenius algebra of the differential operators constructed in
\cite{MMN3}. This algebra is induced by the differential operators associated with arbitrary Young diagrams and
bipartite graphs. The operator associated with the Young diagram
of the transposition of two elements coincides with the cut-and-join operator generating
a differential relation for the generating function of the classical Hurwitz numbers \cite{GJ1}. The operators
associated with the simplest bipartite graphs are found in \cite{N3}. We prove that all the operators associated
with Young diagrams and bipartite graphs give rise to analogous cut-and-join differential relations.

\vspace{2ex}

\paragraph{{\bf Acknowledgements.}}

S.N. is grateful to MPIM and IHES for the kind hospitality and support.

Our work is partly supported by Ministry of Education and Science of
the Russian Federation under contract 8498,
by Russian Federation Government Grant No. 2010-220-01-077, ag.no.11.G34.31.0005,
by NSh-3349.2012.2 (A.Mir. and A.Mor.) and 8462.2010.1 (S.N.),
by CNPq 400635/2012-7, the Brazil National Counsel of Scientific and
Technological Development (A.Mor.), by the program  of UFRN-MCTI (Brazil) (A.Mir.),
by RFBR grants 10-01-00536 (A.Mir. and A.Mor.), 10-01-00678(A.N.) and
by joint grants 11-02-90453-Ukr, 12-02-92108-Yaf-a,
11-01-92612-Royal Society.

\vspace{2ex}

\section{Hurwitz numbers of seamed surfaces}
\vspace{2ex}

\subsection{Seamed surfaces}

We name {\it seamed graph} a one-dimensional topological space
$\Delta$ without boundaries with a finite set of marked points
$\Delta_v\in\Delta$ called {\it vertices} such that the set
$\Delta-\Delta_v$ is a disjoint set of intervals.
These intervals are called {\it edges} of the graph.
One requires that any connected component of the graph is not a point.

{\it Morphism} of a seamed graph $\Delta^1$ into a seamed graph $\Delta^2$
is a continuous map $f:\Delta^1\to\Delta^2$ which is a local homomorphism on the edges of the graph,
such that $f^{-1}(\Delta^2_v)=\Delta^1_v$.

\vspace{1ex}

We mean by \textit{surface} a compact surface, maybe with the boundary and/or non-oriented.

{\it Seamed surface} \cite{AN2} is the set
$(\Omega,\Delta,\varphi)$, where $\Omega$ is a disjoint set of surfaces,
$\Delta$ is a seamed graph and
$\varphi:\partial\Omega\to\Delta$ is an epimorphic morphism of the seamed graphs such that
the restriction
$\varphi$ onto every boundary contour $c\subset\partial\Omega$ is a covering over $\varphi(c)$, and
the below condition $(*)$ is fulfilled.

Identifying the points from $\partial\Omega$ that have same images in $\Delta$, one glues of $\Omega$
the stratified topological space $\Omega_\varphi$, its open two-dimensional strata being homeomorphic to
the two-dimensional strata of $\Omega$, its one-dimensional strata coinciding with the edges of the graph$\Delta$
and its zero-dimensional strata being the vertices of $\Delta$ and the internal special points
of the surface $\Omega$. It is necessary that

{\ }

(*) the punctured vicinity of every vertex of the graph $\Delta$ in the space
$\Omega_\varphi$ is connected.

{\ }

{\it Isomorphism of the seamed surfaces}
$(\Omega, \Delta,\varphi)$ and $(\Omega',\Delta',\varphi')$ is called the homeomorphism
$f:\Omega_\varphi\to\Omega'_{\varphi'}$ inducing the isomorphism of the seamed graphs
$\Delta$ and $\Delta'$. The isomorphism of the seamed surface onto itself is called {\it automorphism}.

\subsection{Covering surfaces with seamed surfaces}

{\it Covering of degree $n$ of the surface} $\Omega$ with use of the seamed surface $(\tilde\Omega,\tilde\Delta,\tilde\varphi)$
is defined to be a continuous map
$f:\tilde\Omega_{\tilde\varphi}\rightarrow\Omega$ such that:

1) $f(\tilde\Omega_{\tilde\varphi})=\Omega$;

2) $f$ maps the seamed graph onto the boundary $\partial\Omega$ of the surface $\Omega$ giving rise to a local
homeomorphism of edges of the seamed graph;

3) there exists only a finite set $\Omega_a$ of points $p\in\Omega\setminus\partial\Omega$ such that the number of their pre-images
is not $n$.

Points of the set $\Omega_a$ are called \textit{internal critical values}. The image $f^{-1}(u)$ of a small contour $u$
surrounding the point $p\in\tilde\Omega_a$ decomposes into simple contours $\tilde u_1,...,\tilde u_m$.

\textit{Topological type} of the internal critical value $p$ of the covering
$f$ is the non-ordered set of numbers
$\alpha=(n_1,...,n_m)$, $n_i$ being the degree of restriction of the map $f$ onto $\tilde u_i$. The sum of all
$n_i$ is equal to $n$. Hence, one may treat $\alpha$ as the Young diagram of degree $n$. The group of automorphisms
$\Aut(\alpha)$ consists of auto-homeomorphisms of the set $\tilde u_1\cup...\cup\tilde u_m$ which are permutable with the
covering $f$.

\textit{From now on, we assume that the boundary of the surface $\Omega$ is oriented.}

Determine the topological type of the boundary point $b\in\partial\Omega$. Consider a small segment
$l\subset\Omega$ surrounding the critical value $q$, i.e. the segment entirely except for
the end-points consisting of the internal
points of the surface and isolating from $\Omega$ a disk containing the point $q$.

The boundary orientation orders the segment end-points $v_1$ and $v_2$. The segment pre-image
$f^{-1}(v)$ forms \textit{a bipartite graph of degree $n$}, i.e. a graph with vertices divided into two sets
ordered in accordance with orientation of the boundary $\partial\Omega$:
$V_1=f^{-1}(v_1)$ and $V_1=f^{-1}(v_2)$, and  with $n$ edges $E$ connecting vertices from the different sets.

The class of topological equivalence of the bipartite graph $(V_1,E,V_2)$ is called
\textit{topological type of the boundary value $q$}. The values of $q$ are \textit{critical} if at least one of the connected
components of the graph $(V_1,E,V_2)$ has more than 2 vertices. The set of the boundary critical values is contained in the set
$\Omega_b$ of pre-images of vertices of the seamed graph.

{\it Isomorphism} of the coverings
$f:\tilde\Omega_{\tilde\varphi}\rightarrow\Omega$ and
$f':\tilde\Omega'_{\tilde\varphi'}\rightarrow\Omega$ of the seamed surfaces
$(\tilde\Omega,\tilde\Delta,\tilde\varphi)$ and
$(\tilde\Omega',\tilde\Delta',\tilde\varphi')$ is the isomorphism of the seamed surfaces
$F:\Omega_\varphi\to\Omega'_{\varphi'}$ such that
$f'F=f$. If $(\tilde\Omega,\tilde\Delta,\tilde\varphi)=$
$(\tilde\Omega',\tilde\Delta',\tilde\varphi')$, then the isomorphism is called {\it automorphism}.
Automorphisms of the covering $f$ forms a group $\Aut(f)$. The groups of automorphisms
of isomorphic coverings are isomorphic.

\vspace{2ex}

\subsection{Hurwitz numbers}

Denote through $\mathcal A_n$ the set of Young diagrams $\Delta$ of degree $|\Delta|=n$. Denote through
$A_n$ the vector space induced by $\mathcal A_n$. Denote through $\mathcal B_n$ the set of isomorphism classes of
the bipartite graphs of degree $|(V_1,E,V_2)|=|E|=n$. Denote through
$B_n$ the vector space induced by $\mathcal B_n$.

Consider the triple $(\Omega,\Omega_a,\Omega_b)$ consisting of the surface $\Omega$,
of the finite set of its internal points $\Omega_a$ and of the finite set of its boundary points $\Omega_b$.
Put $V_{\Omega}=(\bigotimes\limits_{p\in\Omega_a}A_p) \bigotimes(\bigotimes\limits_{q\in\Omega_b}B_q)$,
$A_p$ and $B_q$ being copies of the vector spaces $A_n$ and $B_n$ accordingly.

\vspace{1ex}

Consider the maps
$\alpha :\Omega_a\to\mathcal A_n$, $\beta :\Omega_b\to
\mathcal B_n$ and put $\alpha_p=\alpha(p)$, $\beta_q=
\beta(q)$. Denote through $Cov(\Omega,\alpha,
\beta)$ the set of isomorphism classes of covering of the surface
$\Omega$ by the seamed surfaces $(\tilde\Omega,
\tilde\Delta,\tilde\varphi)$, with the critical values contained in
$\Omega_a\cup\Omega_b$ and the local invariants $\alpha_p\in\mathcal A_n$, $\beta_q\in\mathcal B_n$ at the points
$p\in\Omega_a$, $q\in\Omega_b$.

Following \cite{AN} we call the number
$$H(\Omega,\alpha,\beta) =
\sum_{[f]\in Cov(\Omega,\alpha,\beta)}
1/|\Aut([f])|.$$
{\it Hurwitz number of degree $n$}. For  $\Omega_b=\emptyset$ these numbers coincide
with the standard Hurwitz numbers \cite{H}, \cite{D}.

The Hurwitz numbers give rise to the family of linear functionals
$\mathcal {H}= \{\Phi_{\Omega}:
V_\Omega\to \mathbb{R}\}$, which (in accordance with \cite{AN1},\cite{AN2}) do not depend on the orientation of the
boundary $\partial\Omega$ and form a Klein topological field theory in the meaning of \cite{AN}.

\vspace{2ex}

\subsection{Universal and asymptotic Hurwitz numbers}

We call the Young diagram $\tilde{\Delta}$ of degree $n$ obtained from the Young diagrams $\Delta$ of degree $m\leq n$ by
adding $n-m$ lines of unit length \textit{standard extension of degree $n$} of the Young diagram $\Delta$.
Put  $\rho_n^A(\Delta)= \frac{|\Aut(\tilde{\Delta})|}{|\Aut(\Delta)||\Aut(\tilde{\Delta}
\setminus\Delta)|}\tilde{\Delta}$.
The correspondence $\Delta\mapsto\rho_n^A(\Delta)$ gives rise to the homomorphism of vector spaces
$\rho_n^A: A_m\rightarrow A_n$.

\vspace{1ex}

We call the graph with all connected components having two vertices \textit{simple}.
The graph $\tilde{\Gamma}$ of degree $n$ is called \textit{standard extension of degree $n$}
of the graph $\Gamma$ of degree $m\leq n$ if the complement $\tilde{\Gamma}\setminus\Gamma$
consists of connected components of the graph $\tilde{\Gamma}$ and forms a simple graph.

Denote though $\mathcal{E}_n(\Gamma)$ the set of all standard extensions of degree $n$
of the graph $\Gamma$. Put  $\rho_n^B(\Gamma)= \sum\limits_{\tilde{\Gamma}\in\mathcal{E}_n(\Gamma)}
\frac{|\Aut(\tilde{\Gamma})|}{|\Aut(\Gamma)||\Aut(\tilde{\Gamma}
\setminus\Gamma)|}\tilde{\Gamma}$.
The correspondence $\Gamma\mapsto\rho_n^B(\Gamma)$ gives rise to the homomorphism of vector spaces
$\rho_n^B: B_m\rightarrow B_n$.

\vspace{1ex}

\textit{Free Hurwitz number of degree $n$} is the number
$$H^{fr}_n(\Omega,\{\alpha_p\},\{\beta_q\}) = H(\Omega,\{\rho_n^A(\alpha_p)\}, \{\rho_n^B(\beta_q)\})$$

\vspace{1ex}

The free Hurwitz number $H^{fr}_n(\Omega,\{\alpha_p\},\{\beta_q\})$ whose degree is equal to the maximum of degrees of the
Young diagrams $\{\alpha_p\}$ and bipartite graphes $\{\beta_q\}$ is called \textit{universal Hurwitz number}.

The infinite set of the free Hurwitz numbers
$$H^{as}(\Omega,\{\alpha_p\},\{\beta_q\}) = (H^{fr}_1(\Omega,\{\alpha_p\},\{\beta_q\}),
H^{fr}_2(\Omega,\{\alpha_p\},\{\beta_q\}),\dots)$$
is called \textit{asymptotic Hurwitz number}.

\vspace{2ex}

\section{Algebra of asymptotic Hurwitz numbers}\label{s3}

\vspace{2ex}

\subsection{Algebra of Hurwitz numbers of sphere}\label{s3.1}

Consider a sphere $S$ with two marked points $S_a=\{p_1,p_2\}$.
Then, the equality
$$(\alpha_1,\alpha_2)_{A}=H(S,\{\alpha_1,\alpha_2\})=
\frac{\delta_{\alpha_1,\alpha_2}}{|\Aut(\alpha_1)|},\ \ \ \ \ \alpha_1,\alpha_2\in\mathcal{A}_n$$
gives rise to a non-degenerated symmetric bilinear form $(.,.)_{A}:A_n\times A_n\rightarrow\mathbb{C}$.

Consider a sphere $S$ with three marked points $S_a=\{p_1,p_2,p_3\}$. Then, the equality
$$(\alpha_1\circ \alpha_2,\alpha_3)_{A}= H(S,\{\alpha_1,\alpha_2,\alpha_2\}),\ \ \ \ \
\alpha_1,\alpha_2,\alpha_3\in\mathcal{A}_n$$
gives rise to a binary operation $*:A_n\times A_n\rightarrow A_n$ which
makes of $A_n$ \textit{a commutative Frobenius algebra} \cite{D}.

The diagram $\mathfrak{e}_n^A\in\mathcal{A}_n$ with the lines of unit lengths is the identity element of this algebra.
Besides, $(\alpha_1,\alpha_2)_{A}=l_{A}(\alpha_1\circ\alpha_2)$, where $l_{A}:A_n\rightarrow\mathbb{C}$ is a linear
functional equal to $\frac{1}{n!}$ on $\mathfrak{e}_n^A$ and vanishing on all other Young diagrams.

\vspace{1ex}

In accordance with \cite{D}, the Hurwitz number corresponding to sphere $S$ is given by
$$H(S,\{\alpha_1,\dots,\alpha_k\})= l_{A}(\alpha_1*\dots*\alpha_k
), \ \ \alpha_1,\dots,\alpha_k\in\mathcal{A}_n.$$

\vspace{2ex}

A counterpart of the algebra $A_n$ in the infinite-dimensional case is the Ivanov-Kerov algebra $A_\infty$.
This algebra is the center of the semi-group algebra of the semi-group
$D_\infty$ (\cite{IK}). The semi-group $D_\infty$ consists of pairs $(d,\sigma)$, where $d$
is a set of the natural numbers $\mathbb{N}$ and $\sigma:d\rightarrow d$ is a permutation.
Multiplication is given by the formula
$(d_1,\sigma_1)(d_2,\sigma_2)=(d_1\cup d_2,\sigma_{1}\sigma_{2})$.

The conjugation class of the permutation
$(d,\sigma)$ in the group $S_{\infty}$ of finite permutations of the natural numbers $\mathbb{N}$ is described by
the Young diagram $[d,\sigma]$ of degree $|d|$. Associate with the Young diagram $\Delta$ the orbit
$[\Delta]=\{(d,\sigma)|[d,\sigma]=\Delta\}$.
The sums  $\sum\limits_{(d,\sigma)\in[\Delta]}(d,\sigma)$, where $\Delta\in\mathcal{A}= \bigcup\limits_{n=1}^{\infty}\mathcal{A}_n$
gives rise to the center $A_\infty$ of the semi-group algebra of the semi-group $D_\infty$.
Denote through $\circ$ the multiplication in $A_\infty$.
\vspace{1ex}

Introduce a multiplication $\Delta_1*_n\Delta_2= \rho^A_{n}(\Delta_1) *\rho^A_{n}(\Delta_2)$ between two Young diagrams
$\Delta_1,\Delta_2$ of degree not higher than $n$. Then, in accordance with \cite{MMN2}

$\Delta_1\circ\Delta_2=\sum\limits_ {n=\max\{|\Delta_1|,|\Delta_2|\}} ^{\infty}\{\Delta_1\Delta_2\}_{n}$,
where
$\{\Delta_1\Delta_2\}_n = \Delta_1*_n
\Delta_2$ at $n=\max\{|\Delta_1|,|\Delta_2|\}$
and
\noindent
$\{\Delta_1\Delta_2\}_n = \Delta_1*_n
\Delta_2-\sum\limits_ {k=\max\{|\Delta_1|,|\Delta_2|\}}^{n-1}
\rho^A_n(\{\Delta_1\Delta_2\}_k)$ at $n>\max\{|\Delta_1|,|\Delta_2|\}$.

At the same time, $\{\Delta_1\Delta_2\}_n =0$ at
$n>|\Delta_1|+|\Delta_2|$.

\vspace{2ex}

Let us give on the set of sequences $A^{as}=\{(a^{1},a^{2},\dots)|a^{n}\in A_n\}$,
the binary operation $(a^{1}_1,a^{2}_1,\dots)*(a^{1}_2,a^{2}_2,\dots)=
(a^{1}_1* a^{1}_2,a^{2}_1* a^{2}_2,\dots)$.

Associate with the element $a\in A_n$ the sequence of numbers $a^{as}= (a^{1},a^{2},\dots)\in A^{as}$, where
$a^{i}=0$ at $i<n$ and $a^{i}=\rho^A_{i}(a)$ at $i\geq n$.

\begin{theorem} \label{t3.1} The correspondence $a\mapsto a^{as}$ gives rise to the monomorphism
of algebras $\rho^A_{\uparrow}: A_\infty\rightarrow A^{as}$.
\end{theorem}

\begin{proof} One suffices to check the statement of the theorem for the Young diagrams.
In this case, the claim follows from the identity $\Delta_1*_n
\Delta_2 = \{\Delta_1\Delta_2\}_n + \sum\limits_ {k=\max\{|\Delta_1|,|\Delta_2|\}}^{n-1}
\rho^A_n(\{\Delta_1\Delta_2\}_k)$
\end{proof}

Thus, the algebra $A^{as}$ is isomorphic to the algebraic closure of the algebra $A_\infty$, or, more precisely,
to the algebraic closure of the algebra $A_\uparrow=\rho_\uparrow(A_\infty)$.

Consider the linear operator $l_{A}^{as}: A^{as}\rightarrow\mathbb{C}^{\infty}= \{(c_1,c_2,\dots)|c_i\in\mathbb{C}\}$, where
$l_{A}^{as}(a^{1},a^{2},\dots)=
(l_{A}(a^{1}),l_{A}(a^{2}),\dots)$ and the bilinear operator $(a_1,a_2)_{A}: A^{as}\times A^{as}\rightarrow\mathbb{C}^{\infty}$
where $(a_1,a_2)_{A}=l_{A}^{as}(a_1* a_2)$.
Then, from theorem \ref{t3.1} follows

\begin{theorem} \label{t3.2} The multiplication in the algebra $A^{as}_{\uparrow}$ is determined by the equality
$$(a_1* a_2,a_3)_A=H^{as}(S,\{a_1,a_2,a_3\}).$$
Moreover,
$$H^{as}(S,\{a_1,\dots,a_k\})= l_{A}^{as}(a_1*\dots* a_k), \ \ a_1,\dots,a_k\in A^{as}_{\uparrow}$$
\end{theorem}

\vspace{2ex}

Associate with the element $a\in A_\infty$ the sum $l_{A}^{\Sigma}(a)=\sum\limits_{n=0}^{\infty}l_{A}(a^n)$, where
$(a^{1},a^{2},\dots)=\rho_{\uparrow}^A(a)$

\begin{lemma}\label{l3.2} The sum $l_{A}^{\Sigma}(\rho_{\uparrow}(a))$ absolutely converges for any
$a\in A_{\infty}$, and $l_{A}^{\Sigma}\rho_{\uparrow}=el_{A}$.
\end{lemma}

\begin{proof} One suffices to check the statement of the lemma for the Young diagrams. In this case,
$l_{A}(\Delta)=0$ unless $\Delta=\mathfrak{e}_n^A$, when $l_{A}(\mathfrak{e}_n^A)=
\sum\limits_{k=n}^{\infty}\frac{k!}{k!n!(k-n)!}= \frac{1}{n!}\sum\limits_{k=n}^{\infty}\frac{1}{(k-n)!}=\frac{e}{n!}$.
\end{proof}

\vspace{2ex}

\subsection{Algebra of boundary Hurwitz numbers of disk}\label{s3.2}

\vspace{1ex}

We understand by \textit{boundary Hurwitz numbers} the Hurwitz numbers of coverings over the disk without internal critical values.

\vspace{1ex}

Denote through $\star:B_n\rightarrow B_n$ the involution induced by changing the
orientation of the graphs: $\star(V_1,E,V_2)=(V_1,E,V_2)^{\star}=(V_2,E,V_1)$.

Consider a disk $D$ with two marked points $D_b=\{q_1,q_2\}$.
Then, the equality
$$(\beta_1,\beta_2)_{B}=H(D,\{\beta_1,\beta_2\})=
\frac{\delta_{\beta_1,\beta_2^{\star}}}{|\Aut(\beta_1)|},\ \ \ \ \ \beta_1,\beta_2\in\mathcal{B}_n$$
gives rise to a non-degenerated symmetric bilinear form $(.,.)_{B}:B_n\times B_n\rightarrow\mathbb{C}$.

Consider a disk $D$ with three marked points $D_b=\{q_1,q_2,q_3\}$. Then, the equality
$$(\beta_1*\beta_2,\beta_3)_{B}= H(D,\{\beta_1,\beta_2,\beta_2\}),\ \ \ \ \ \beta_1,\beta_2,\beta_3\in\mathcal{B}_n$$
gives rise to a binary operation $B_n\times B_n\rightarrow B_n$ which makes of $B_n$ a generically
\textit{non-commutative Frobenius algebra} \cite{AN2}.

\vspace{1ex}

Following \cite{AN1},\cite{AN2}, describe the algebra $B_n$ in terms of graphs.
It is induced by the set of bipartite graphs $\mathcal{B}_n$ of degree $n$.
Multiplication is defined as follows. Let $(V_1,E,V_2)$ и $(V'_1,E',V'_2)$ be a pair of bipartite graphs
with $n$ edges. Denote through $Hom(V_2,V'_1)$ the set of maps $\chi:V_2\rightarrow V'_1$
which preserve the valency of vertices. Associate with each map the bipartite graph
$(V_2,E_\chi,V'_1)$ with the edges connecting only the vertices $v$ and $\chi(v)$, where $v\in V_2$,
the number of edges connecting $v$ and $\chi(v)$ being equal to the valency of the vertex $v$.

Call the subset $F\subset E\times E'$ consistent with $\chi$, if the restriction
of the natural projections $E\times E'\rightarrow
E$, $E\times E'\rightarrow E'$
onto $F$ are in one-to-one correspondence and
$\chi(V_2(e))=V_1(e')$ for any $(e,e')\in F$. Denote through
$M_\chi$ the set of such $F$'s. Associate with the subset $F\in
M_\chi$ the bipartite graph $(V_1,\overline{F},V'_2)$, its edges being pairs of edges $(e,e')\in F$ glued together
at the points
$V_2(e)$ and $V_1(e')$. Denote through
$\Aut_F(V_1,\overline{F},V'_2)\subset\Aut(V_1,\overline{F},V'_2)$ the subgroup consisting of the automorphisms
inducing on the set $E$ the automorphism of the graph
$(V_1,E,V_2)$.

Now construct the map $\mathcal B_n\times\mathcal
B_n\rightarrow B_n$ by putting $[(V_1,E,V_2)]*[(V'_1,E',V'_2)]=$
$\sum_{\chi\in Hom(V_2,V'_1)}\sum_{F\in
M_\chi}\frac{|\Aut((V_2,\overline{F},V'_1))|}
{|\Aut_F((V_1,\overline{F},V'_2))|}[(V_1,\overline{F},V'_2)]$.
Continuing it by linearity, one obtains the binary operation
that makes of $B_n$ an algebra.

This operation has a simple geometrical meaning. The product $[(V_1,E,V_2)]*[(V'_1,E',V'_2)]$ is contributed by
the valency-preserving identifications of vertices from $V_2$ with those from $V'_1$. As a result of such identification,
there emerges ``a singular graph" with vertices $V_1\cup V'_2$ and edges intersecting on ``the set of singularities"
$V_2=V'_1$. The product is defined to be a linear combination of ``resolutions" of these singularities, i.e.
pairwise gluing of the edges coming into the common vertex from $(V_1,E,V_2)$ and
$(V'_1,E',V'_2)$, which induces a bipartite graph.

\vspace{1ex}

The sum $\mathfrak{e}_n^B= \sum\limits_{E\in\mathcal{E}_n}\frac{E}{|\Aut(E)|}$ over the set of all simple graphs
of degree $n$ is an identity element of the algebra $B_n$.
Besides, $(\beta_1,\beta_2)_{B}= l_{B}(\beta_1*\beta_2)$, where $l_{B}:B_n\rightarrow\mathbb{C}$ is
the linear functional equal to $\frac{1}{|\Aut(\Gamma)|}$ on the simple graphs $\Gamma$ and vanishing
on all other graphs.

\vspace{1ex}

In accordance with \cite{AN2} the Hurwitz number corresponding to the disk with marked boundary points is given by the
formula
$$H(D,\{\beta_1,\dots,\beta_k\})=l_{B}(\beta_1*\dots*\beta_k),\ \ \beta_1,\dots,\beta_k\in\mathcal{B}_n.$$

\vskip 0.6cm

We use the algebras $B_n$ of the bipartite graphs of fixed degree in order to construct an algebra on the vector space
$B_\infty$ generated by all bipartite graphs.

\vspace{1ex}

Introduce multiplication of the graphs $\Gamma_1,\Gamma_2$ of degree not greater than $n$:
$\Gamma_1*_n\Gamma_2=\rho^B_{n}(\Gamma_1)*\rho^B_{n}(\Gamma_2)$.

Define the binary operation on the vector space $B_\infty$ requiring that

$\Gamma_1\circ\Gamma_2=\sum\limits_ {n=\max\{|\Gamma_1|,|\Gamma_2|\}} ^{\infty}\{\Gamma_1\Gamma_2\}_{n}$,
where
$\{\Gamma_1\Gamma_2\}_n = \Gamma_1*_n
\Gamma_2$ at $n=\max\{|\Gamma_1|,|\Gamma_2|\}$
and
\noindent
$\{\Gamma_1\Gamma_2\}_n = \Gamma_1*_n
\Gamma_2-\sum\limits_ {k=\max\{|\Gamma_1|,|\Gamma_2|\}}^{n-1}
\rho^B_n(\{\Gamma_1\Gamma_2\}_k)$ at $n>\max\{|\Gamma_1|,|\Gamma_2|\}$.

\begin{lemma} \label{l3.2b} The algebra $B_\infty$ is associative and  $\{\Gamma_1\Gamma_2\}_n =0$ at
$n>|\Gamma_1|+|\Gamma_2|$.
\end{lemma}

\begin{proof}  It follows from the definitions that $\Gamma_1\circ\Gamma_2\circ\Gamma_3=\sum\limits_
{n=\max\{|\Gamma_1|,|\Gamma_2|,
|\Gamma_3|\}}^{\infty}\{\Gamma_1\Gamma_2\Gamma_3\}_{n}$,
\noindent
where $\{\Gamma_1\Gamma_2\Gamma_3\}_n\subset B_n$, $\{\Gamma_1\Gamma_2\Gamma_3\}_{n}$ being not changing
under any permutations of $\Gamma_i$. It follows associativity of the algebra $B_{\infty}$. The equality $ \Gamma_1*_n
\Gamma_2=\sum\limits_ {k=\max\{|\Gamma_1|,|\Gamma_2|\}}^{n-1}
\rho^B_n(\{\Gamma_1\Gamma_2\}_k)$ at $n>\max\{|\Gamma_1|,|\Gamma_2|\}$ follows from the description of the product $*$
above.
\end{proof}

\vskip 0.6cm

Define on the set of sequences $B^{as}=\{(b^{1},b^{2},\dots)|b^{n}\in B_n\}$,
a binary operation $(b^{1}_1,b^{2}_1,\dots)*(b^{1}_2,b^{2}_2,\dots)=
(b^{1}_1* b^{1}_2,b^{2}_1* b^{2}_2,\dots)$.

Associate with the element $b\in B_n$ the sequence of numbers $b^{as}= (b^{1},b^{2},\dots)\in B^{as}$, where
$b^{i}=0$ at $i<n$ and $b^{i}=\rho^B_{i}(a)$ at $i\geq n$.

\begin{theorem} \label{t3.3} The correspondence $b\mapsto b^{as}$ gives rise to the monomorphism of algebras
$\rho_{\uparrow}^B: B_\infty\rightarrow B^{as}$.
\end{theorem}

\begin{proof} One suffices to check the claim of the lemma for graphs. In this case, it follows from the equality
$\Gamma_1*_n
\Gamma_2 = \{\Gamma_1\Gamma_2\}_n + \sum\limits_ {k=\max\{|\Gamma_1|,|\Gamma_2|\}}^{n-1}
\rho^B_n(\{\Gamma_1\Gamma_2\}_k)$
\end{proof}

Thus, the algebra $B^{as}$ is isomorphic to the algebraic closure
of the algebra $B_\infty$, more precisely, to the algebraic closure
of the algebra $B_\uparrow=\rho_\uparrow(B_\infty)$.

Consider the linear functional  $l_{B}^{as}:  B^{as}\rightarrow\mathbb{C}^{\infty}$, where
$l_{B}^{as}(b^{1},b^{2},\dots)=(l_{B}(b^{1}),l_{B}(b^{2}),\dots)$ and the bilinear operator
$(b_1,b_2)_{B}: B^{as}\times B^{as}\rightarrow\mathbb{C}^{\infty}$
where $(b_1,b_2)_{B}=l_{B}^{as}(b_1* b_2)$.
Then, from theorem \ref{t3.3} it follows that

\begin{theorem} \label{t3.4} The multiplication in the algebra $B_{\uparrow}$ is defined by the equality
$$(b_1* b_2,b_3)_B=H^{as}(D,\{b_1,b_2,b_3\}).$$
Moreover,
$$H^{as}(D,\{b_1,\dots,b_k\})= l_{B}(b_1*\dots* b_k)$$
\end{theorem}

\vspace{2ex}

\subsection {Cardy-Frobenius algebras}\label{s3.3}

Following \cite{AN}, \cite{AN2}, remind the definition of the (finite-dimensional) equipped Cardy-Frobenius algebra.

\vspace{1ex}

\textit{Frobenius pair} is a set $(C,l^C)$ that consists of a finite-dimensional associative algebra
with an identity element $C$ and a linear functional $l^C:C\rightarrow\mathbb{C}$ such that
the bilinear form $(c_1,c_2)_C= l^C(c_1c_2)$ induced by the functional is non-degenerate.

\textit{Casimir element} of the Frobenius pair $(C,l^C)$ is the element
$K_C=\sum\limits_{i=1}^{n}F^{ij}e_ie_j\in C$, where $\{e_1,\dots,e_n\}$ is the basis of the space $C$ and $\{F^{ij}\}$
is the matrix inverse to the matrix $(e_i,e_j)_C$.

\vspace{1ex}

For the Frobenius pairs $(A,l^A)$ $(B,l^B)$ and the linear operator $\phi: A\rightarrow B$ denote through
$\phi^*: B\rightarrow A$ the linear operator defined by the condition $(\phi^*(b),a)_A=(b,\phi(a))_A$.

\vspace{2ex}

\textit{Cardy-Frobenius algebra} is the set $((A,l^A),(B,l^B),\phi)$ which consists of

1) a commutative Frobenius pair $(A,l^A)$;

2) an arbitrary Frobenius pair $(B,l^B)$;

3) a homomorphism of algebras $\phi: A\rightarrow B$ such that the image $\phi(A)$ belongs to the center of the algebra
$B$ and $(\phi^*(b'),\phi^*(b''))_A =\tr K_{b'b''}$, where the operator $K_{b'b''}:B\rightarrow B$ is defined by
$K_{b'b''}(b)=b'bb''$.

\vspace{1ex}

\textit{Equipped Cardy-Frobenius algebra} is the set $((A,l^A),(B,l^B),\phi,U,\star)$ which consists of

1) the Cardy-Frobenius algebra $((A,l^A),(B,l^B),\phi)$;

2) involutive anti-automorphisms $\star :A\rightarrow A$ and
$\star : B\rightarrow B$ such that $l^A(x^\star)=l^A(x)$, $l^B(x^\star)=l^B(x)$, $\phi(x^\star)=\phi(x)^\star$;

3) the element $U\in A$ such that $U^2=K_A^\star$ and $\phi(U)=K_B^\star$.

The commutative Frobenius pairs are in one-to-one correspondence \cite{D} with closed topological field theories
in the sense of \cite{At}. It is classical and regular Hurwitz numbers that lead to these theories
\cite{D}.

The Cardy-Frobenius algebras are in one-to-one correspondence \cite{AN} with open-closed topological field theories.
An example of the Cardy-Frobenius algebras is, in particular, the algebras corresponding to open-closed string theory
\cite{Moore}, \cite{Laz}, \cite{MS}.

The equipped Cardy-Frobenius algebras are in one-to-one correspondence \cite{AN} with the Klein topological theories, i.e.
topological theories also determined on non-oriented surfaces \cite{AN}. These theories are associated with
the Hurwitz numbers of the seamed surfaces \cite{AN1}, \cite{AN2} and with the regular
Hurwitz numbers of the seamed surfaces \cite{AN3}.

Every real representations of the finite group induces a semi-simple equipped Cardy-Frobenius algebra \cite{LN}.
There exists a complete classification of the semi-simple equipped Cardy-Frobenius algebras \cite{AN}.

\vspace{2ex}

The above definitions require inverting matrices. Hence, their extension to the infinite-dimensional case requires
an additional care. We additionally demand that the algebras can be presented as direct (Cartesian) products of
finite-dimensional algebras
$A=\prod\limits_{\gamma\in\mathfrak{C}}A_\gamma$, $B=\prod\limits_{\gamma\in\mathfrak{C}}B_\gamma$ and, instead of functionals
on $A$ and $B$, we will consider the families of functionals
$l^A=\{l^A_\gamma:A_\gamma\rightarrow\mathbb{C}\}$, $l^B=\{l^B_\gamma:B_\gamma\rightarrow\mathbb{C}\}$ such that:

1)$(A_\gamma,l^A_\gamma)$ and $(B_\gamma,l^B_\gamma)$ are the Frobenius pairs ;

2)$\phi(A_\gamma)\in B_\gamma$ and the restrictions $\phi_\gamma$ of the homomorphism  $\phi$ onto $A_\gamma$
give rise to the Cardy-Frobenius algebras $((A_\gamma,l^A_\gamma),(B_\gamma,l^B_\gamma),\phi_\gamma)$;

3) The involution $\star$ preserves the subalgebras $A_\gamma$, $B_\gamma$ and, along with the projections
$U_\gamma$ of the element $U\in A$ onto $A_\gamma$, gives rise to the equipped Cardy-Frobenius algebras
$((A_\gamma,l^A_\gamma),(B_\gamma,l^B_\gamma),\phi_\gamma,U_\gamma,\star)$.

\vspace{1ex}

\subsection{Algebra of asymptotic Hurwitz numbers}\label{s3.4}

As was already noted, the sets $((A_n,l_A)$ and $(B_n,l_B)$ form Frobenius pairs. Besides, in
\cite{AN2} there was constructed the homomorphism
$\phi_n: A_n\rightarrow B_n$ and the element $U_n$ such that the set $((A_n,l_A),(B_n,l_B),\phi_n,U_n,\star)$
is the equipped Cardy-Frobenius algebra at any $n$.

On the other hand, $A^{as}=\prod\limits_{\gamma=1}^\infty A_n$ and $B^{as}=\prod\limits_{\gamma=1}^\infty B_n$.
The families  $\{\phi_n\}$ and  $\{U_n\}$ give rise to the homomorphism $\phi^{as}: A^{as}\rightarrow B^{as}$ and the
element $U^{as}\in A^{as}$. Thus, the set $((A^{as},l_A^{as}),(B^{as},l_B^{as}),\phi^{as}, U^{as},\star)$ also
forms the equipped Cardy-Frobenius algebra.

In accordance with theorems \ref{t3.1} and \ref{t3.3} the algebras  $A^{as}$ and $B^{as}$
are isomorphic to the algebraic closures of the algebras $A_\infty$ and $B_\infty$. Moreover, there is

\begin{theorem} If $a\in A_n$, then $\phi^{as}\rho_{\uparrow}^A(a)=\rho^B_{\uparrow}\phi_n(a)$.
\end{theorem}

\begin{proof} The theorem is equivalent to the relation
$\phi_{n+1}\rho_{n}^A(a)=\rho^B_{n}\phi_n(a)$, where $\phi_n(a)$ is defined in accordance with \cite{AN2}, by the
relation $H(D,\phi_n(a),b)= H(D,a,b)$ for all $b\in B_n$. On the other hand, it follows from the definition of the Hurwitz numbers
that $H(D,\rho^B_{n}(\phi_n(a)),b')= H(D,\rho_{n}^A(a),b')$ for all $b'\in B_{n+1}$, if $H(D,\phi_n(a),b)= H(D,a,b)$ for
all $b\in B_n$.
\end{proof}

Thus,
\begin{corollary} The algebraic closure of the structure $((A_\infty,l_A),(B_\infty,l_B),\{\phi_n\},\{U_n\},\star)$
forms the equipped Cardy-Frobenius algebra.
\end{corollary}

\section{Differential equations for generating functions} \label{s4}

\vspace{2ex}

\subsection{Cut-and-join operators}

Now we construct a representation of the algebras $A$ and $B$ as algebras of differential operators on the space of
functions of infinitely many variables $\{X_{ij}|i,j=1,\dots,\}$ and express in these terms the map
$\phi$.

The algebra $A$ is realized by the algebra of the cut-and-join operators $W(\Delta)$ \cite{MMN1},\cite{MMN2},
\cite{AMMN}.

Remind their construction. We need differential operators of the form
$D_{ab}=\sum_{e=1}^\infty X_{ae}\frac{\partial}{\partial X_{be}}$. Associate with the Young diagram
$\Delta=[\mu_1,\mu_2,\dots,\mu_k]$ with the lines of lengths $\mu_1\geq\mu_2\geq\dots\geq\mu_k$ the numbers
$m_j=m_j(\Delta)= |\{i|\mu_i=j\}|$ and $\kappa(\Delta)= (|\Aut(\Delta)|)^{-1} =(\prod\limits_{j}m_j!j^{m_j})^{-1}$.
Associate with the Young diagram $\Delta$ \textit{cut-and-join operator}
$W(\Delta)=\kappa(\Delta):\prod\limits_{j}(\tr D^j)^{m_j}:$,
where $D$ is the infinite-dimensional matrix with elements
$D_{ab}=\sum_{e=1}^\infty X_{ae}\frac{\partial}{\partial X_{be}}$ and $:...:$ denotes the normal ordering, when
all derivatives are placed to the right of all $X_{ab}$ in the product.
The product of operators is denoted by $\circ$. Denote through $W_\infty$ the algebra induced by the operators
$W(\Delta)$.

\begin{theorem} \cite{MMN2} The correspondence $\Delta\mapsto W(\Delta)$ states the isomorphism
$\varphi^A:A_\infty\rightarrow W$.
\end{theorem}

\vspace{1ex}

\subsection{Graph-operators}

Associate with the monomial $x=X_{a_1b_1}\dots X_{a_nb_n}$ of degree $n$ the bipartite graph $\Gamma(x)$ with edges
$\{E_1,\dots,E_m\}$, where the edges $E_i$ and $E_j$ have common left (accordingly, right) vertex iff
$a_i=a_j$ (accordingly, $b_i=b_j$). Now associate with the graph $\Gamma$ \textit{graph-variable}
$X_{\Gamma}=\frac{1}{|\Aut(\Gamma)|}\sum X$,
where the sum goes over all monomials $x$ such that $\Gamma(x)=\Gamma$. Denote through $X_n$
the vector space generated by the graph-variables of degree $n$.

Associate with the operator $D= :D_{a_1b_1}\dots D_{a_nb_n}:$ the bipartite graph $\Gamma(\mathcal{D})$ with
the edges $\{E_1,\dots,E_m\}$, where the edges $E_i$ and $E_j$ have common left (accordingly, right) vertex iff
$a_i=a_j$ (accordingly, $b_i=b_j$). Now associate with the graph $\Gamma$ the operator
$V[\Gamma]= \frac{1}{|\Aut(\Gamma)|}\sum\mathcal{D}$,
where the sum goes over all operators
$\mathcal{D}$ such that $\Gamma(D)=\Gamma$. We call such operators
\textit{graph-operators}.

\vspace{1ex}

Define action of the graph-operators of degree $n$ on the graph-variables of the same degree. The usual action
of the graph-operators on the graph-variables turns out to be as a linear combination of graph-variables with infinite
coefficients. Hence, to define a correct differentiation we need to introduce some regularization. To this end,
consider, along with the (full) graph-operator and graph-variable $V[\Gamma]$, $X_{[\Gamma']}$
the restricted graph-operator $V^N[\Gamma]$ and graph-variable  $X_{[\Gamma']}^N$ defined similarly to the
full ones, but with the infinite set of variables $\{X_{ij}|i,j= 1,\dots,\}$ replaced with the finite one
$\{X_{ij}|i,j= 1,\dots,N\}$.

Define the action of the graph-operator $V^N[\Gamma]$ on the graph-variable $X_{[\Gamma']}^N$
as the action of the usual differential operator multiplied
by $\frac{(N-|R(\Gamma)|)!}{N!}$. One can easily see that $V^N[\Gamma](X^N_{[\Gamma']})$
is a linear combination of the restricted graph-variables $X^N_{[\Gamma"]}$. Moreover, the coefficients of this
linear combination are the same at any $N>|E(\Gamma)|$.
Now define $V[\Gamma](X_{[\Gamma']})= \lim\limits_{N\rightarrow\infty} V^N[\Gamma](X^N_{[\Gamma']})$.
This operation is naturally continued to $|\Gamma|\ne |\Gamma'|$: $V[\Gamma](X_{[\Gamma']})=0$ at
$|\Gamma|>|\Gamma'|$ and $V[\Gamma](X_{[\Gamma']})= V[\rho_{|\Gamma'|}(\Gamma)](X_{[\Gamma']})$ at $|\Gamma|<|\Gamma'|$.

Denote through $V_\infty$ the algebra generated by the operators $V(\Gamma)$.
Define the operation $\circ$ on $V$ requiring that the operator $V[\Gamma_1]\circ V[\Gamma_2]$
acts on all the graph-variables $X[\Gamma]$ as $V[\Gamma_1](V[\Gamma_2](X[\Gamma]))$.

\vspace{1ex}

\begin{theorem} \cite{MMN3} The correspondence $\Gamma\mapsto V(\Gamma)$ establishes the isomorphism
$\varphi^B:B_\infty\rightarrow V$.
\end{theorem}

\vspace{1ex}

The cut-and-join  operators act on the space of graph-variables by the usual differentiation.
Define the homomorphism of algebras $f: W \rightarrow V$ requiring that the operator $f(w)$
acts on all the graph-variables as the operator $w$ (we prove below that such an operator exists).

\vspace{1ex}

\begin{theorem}\cite{MMN3} $f\varphi^A=\varphi^B\phi$
\end{theorem}

\vspace{2ex}

\subsection{Generating function}

The cut-and-join operators are closely related to special generating functions of the classical Hurwitz numbers
\cite{MMN2}. We construct now the generating function of the Hurwitz numbers for the seamed surfaces which is
related to the graph-operators.

Associate with each Young diagram $\Delta$ and each bipartite graph $\Gamma$ formal variables
$\alpha_\Delta$ и $\beta_\Gamma$.

Fix at the boundary of the disk $D$ a point $q$ and associate with it a bipartite graph $\Gamma$.
Fix at the boundary of the disk pairwise distinct points $q_1,\dots,q_m$ and associate with them
bipartite graphs $\Gamma_1,\dots, \Gamma_m$, where $|\Gamma_i|\leq|\Gamma|$.
Fix pairwise distinct internal disk points $p_1,\dots,p_n$ and associate with them Young diagrams
$\Delta_1,\dots, \Delta_n$, where $|\Delta_i|\leq|\Gamma|$. Denote through
$<\Delta_1,\dots,\Delta_n|\Gamma_1,\dots,\Gamma_m\parallel \Gamma>$ the universal Hurwitz number corresponding to this set of data.
Put $<\Delta_1,\dots,\Delta_n|\Gamma_1,\dots,\Gamma_m\parallel \Gamma>=0$, if the degree of at least one Young diagram or
one graph from the set is larger than $|\Gamma|$.

Part the set of Young diagrams $\Delta_1,\dots,\Delta_n$ into the maximal groups of pairwise coinciding diagrams. Let
$n_1,\dots,n_k$ ($n_1+\dots+n_k=n$) be the numbers of elements in these groups.
Part the set of graphs $\Gamma_1,\dots,\Gamma_m$ into the maximal groups of pairwise coinciding graphs. Let
$m_1,\dots,m_l$ ($m_1+ \dots +m_l=m$) be the numbers of elements in these groups.

Associate with the set of data $(\Delta_1,\dots, \Delta_n|\Gamma_1,\dots,\Gamma_m\parallel\Gamma)$ the monomial
$<\Delta_1,\dots, \Delta_n|\Gamma_1,\dots,\Gamma_m\parallel\Gamma> \frac{\alpha_{\Delta_1}\dots\alpha_{\Delta_n}}
{n_1!\dots n_k!}\frac{\beta_{\Gamma_1}\dots\beta_{\Gamma_m}}
{m_1!\dots m_l!}X_\Gamma$, where $X_\Gamma$ is the graph-variable.

Denote through $Z(\{\alpha_{\Delta}\},\{\beta_{\Gamma}\}
|\{X_\Gamma\})$ the formal sum of all such monomials treated as a function of all variables of a kind of
$\alpha_{\Delta}$, $\beta_{\Gamma}$ and $X_\Gamma$.

Similarly fix now at the boundary of the disk $D$ two distinct points $q$, $q'$ and associate with them bipartite graphs
$\Gamma$, $\Gamma'$, where $|\Gamma'|\leq|\Gamma|$. Fix at the boundary of the disk pairwise distinct points $q_1,\dots,q_m$
lying outside the arc connecting the points $q$, $q'$ and associate with them
bipartite graphs $\Gamma_1,\dots, \Gamma_m$, where $|\Gamma_i|\leq|\Gamma|$.
Fix pairwise distinct internal disk points $p_1,\dots,p_n$ and associate with them Young diagrams
$\Delta_1,\dots, \Delta_n$, where $|\Delta_i|\leq|\Gamma|$.
Denote through $<\Delta_1,\dots, \Delta_n| \Gamma_1,\dots,\Gamma_m| \Gamma'\parallel\Gamma>$
the universal Hurwitz number corresponding to the set of data.
Put  $<\Delta_1,\dots, \Delta_n|\Gamma_1,\dots,\Gamma_m| \Gamma'\parallel\Gamma>=0$,
if the degree of at least one Young diagram or
one graph from the set is larger than $|\Gamma|$.

Part the set of Young diagrams $\Delta_1,\dots,\Delta_n$ into the maximal groups of pairwise coinciding diagrams. Let
$n_1,\dots,n_k$ ($n_1+\dots+n_k=n$) be the numbers of elements in these groups.
Part the set of graphs $\Gamma_1,\dots,\Gamma_m,\Gamma'$ into the maximal groups of pairwise coinciding graphs. Let
$m_1,\dots,m_l$ ($m_1+ \dots +m_l=m+1$) be the numbers of elements in these groups.

Associate with the set of data $(\Delta_1,\dots, \Delta_n|\Gamma_1,\dots,\Gamma_m| \Gamma'\parallel\Gamma)$ the monomial
$<\Delta_1,\dots, \Delta_n|\Gamma_1,\dots,\Gamma_m| \Gamma'\parallel\Gamma>\frac{\alpha_{\Delta_1}\dots\alpha_{\Delta_n}}
{n_1!\dots n_k!}\frac{\beta_{\Gamma_1}\dots\beta_{\Gamma_m}}
{m_1!\dots m_l!}X_\Gamma$, where $X_\Gamma$ is the graph-variable.

Denote through $Z_{\Gamma'}(\{\alpha_{\Delta}\},\{\beta_{\Gamma}\}
|\{X_\Gamma\})$ the formal sum of all such monomials treated as a function of all variables of a kind of
$\alpha_{\Delta}$, $\beta_{\Gamma}$ and $X_\Gamma$.

\begin{theorem}\label{t4.4} The functions $Z$ and  $Z_{\Gamma}$ satisfy the equations
$$\frac{\partial Z} {\partial\alpha_{\Delta}}=W(\Delta)Z$$
$$\frac{\partial Z_{\Gamma'}}{\partial\beta_{\Gamma'}} =V(\Gamma')Z$$
\end{theorem}

\begin{proof}The equality $\frac{\partial Z_{\Gamma'}}{\partial\beta_{\Gamma'}} =V(\Gamma')Z$ is equivalent to the
system of relations between the numbers $<\Delta_1,\dots,\Delta_n|\Gamma_1, \dots,\Gamma_m|
\Gamma'\parallel\Gamma>$ and $<\Delta_1,\dots,\Delta_n|\Gamma_1, \dots,\Gamma_m\parallel\Gamma>$, that is,
 $<\Delta_1,\dots,\Delta_n|\Gamma_1,\dots,\Gamma_m| \Gamma'\parallel\Gamma>=\sum\limits_{i=1}^k
 <\Delta_1,\dots,\Delta_n|\Gamma_1,\dots,\Gamma_m,\Gamma^i>F^{ij} <\Gamma^i\parallel\rho_{|\Gamma|}^B(\Gamma')*\Gamma>$,
where $\{\Gamma^i\}$ is the set of all bipartite graphs of degree $|\Gamma|$ and $F^{ij}$ is the matrix inverse to
$\{<\Gamma^i\parallel\Gamma^j>\}$.
These relations are proved in \cite{AN2} and mean that the Hurwitz numbers are correlators in open-closed topological
field theory.
The relation $\frac{\partial Z}{\partial\alpha_{\Delta}}=
W(\Delta)Z$ is proved analogously.
\end{proof}

\vspace{2ex}

Associate with each connected bipartite graph $\gamma$ a formal variable $q_\gamma$. Consider the algebra $Y$ generated by all
variables $q_\gamma$. The correspondence $q_\gamma\leftrightarrow X_\gamma$ allows one to interpret the arbitrary
graph-variable $X_\gamma$ as the monomial $q_{\gamma_1}\dots q_{\gamma_k}\in Y$, where $\gamma_1,\dots,\gamma_k$
are the connected components of the graph $\Gamma$. The generating functions $Z(\{\alpha_{\Delta}\},\{\beta_{\Gamma}\}
|\{X_\Gamma\})$ and $Z_{\Gamma'}(\{\alpha_{\Delta}\},\{\beta_{\Gamma}\}
|\{X_\Gamma\})$ then becomes generating functions $\mathcal{Z}(\{\alpha_{\Delta}\},\{\beta_{\Gamma}\}
|\{q_\gamma\})$ and $\mathcal{Z}_{\Gamma'}(\{\alpha_{\Delta}\}, \{\beta_{\Gamma}\}|\{q_\gamma\})$. The differential
operators $W(\Delta)$ and $V(\Gamma)$ acting on the space of graph-variables, after the replacement become
differential operators $\verb"W"(\Delta)$, $\verb"V"(\Gamma)$, which act on the algebra $Y$
of variables $\{q_\gamma\}$. Theorem \ref{t4.4} then becomes

\begin{theorem}\label{t4.5} The functions $\mathcal{Z}$ and  $\mathcal{Z}_{\Gamma}$ satisfy the equations
$$\frac{\partial \mathcal{Z}} {\partial\alpha_{\Delta}}=\verb"W"(\Delta)\mathcal{Z}$$
$$\frac{\partial \mathcal{Z}_{\Gamma'}}{\partial\beta_{\Gamma'}} =\verb"V"(\Gamma')\mathcal{Z}$$
\end{theorem}

In simplest cases related to coverings by the Klein surfaces \cite{Nat1} this claim is proved in
\cite{N3} by independent direct combinatorial calculations.

\end{document}